\documentclass[a4paper,prl,twocolumn,showpacs,superscriptaddress,nofootinbib]{revtex4}
\usepackage{amsmath,amscd,amsfonts,amssymb, color,bbm, bm,
	braket, xcolor, setspace, cancel, stmaryrd, calc, amsthm} 
\usepackage{ae}
\usepackage{physics}
\usepackage[T1]{fontenc}
\usepackage{mathptmx} 
\usepackage[cal=cm]{mathalfa}
\usepackage[ansinew]{inputenc}
\usepackage{amsmath}
\usepackage{amssymb}
\usepackage[caption=false]{subfig}
\usepackage{multirow}
\usepackage{array}
\usepackage[]{graphicx}
\usepackage{wrapfig}
\usepackage{makecell}
\usepackage{xparse}
\usepackage{dcolumn}
\usepackage{color}
\usepackage[colorlinks=true, linkcolor=red, citecolor=blue, urlcolor=blue]{hyperref}
\newtheorem{theorem}{Theorem}
\graphicspath{ {./images1/} }

\newtheorem{prop}[theorem]{Proposition}
\newtheorem{defin}{Definition}
\newtheorem{rem}[theorem]{Remark}

\newtheorem{cor}[theorem]{Corollary}

\begin{document}
	\title{Incompleteness is necessary for activation of nonlocality without entanglement}
	\author{Atanu Bhunia}
	\email{atanu.bhunia31@gmail.com}
	\affiliation{Physics and Applied Mathematics Unit, 203 B.T. Road,  Indian Statistical Institute,  Kolkata, 700108, India}
		\author{Saronath Halder}
	\email{saronath.halder@vitap.ac.in}
	\affiliation{Department of Physics, School of Advanced Sciences, VIT-AP University, Beside AP Secretariat, Amaravati 522241, Andhra Pradesh, India}
    \author{Preeti Parashar}
	\email{parashar@isical.ac.in}
	\affiliation{Physics and Applied Mathematics Unit, 203 B.T. Road,  Indian Statistical Institute,  Kolkata, 700108, India}
  \author{Ritabrata Sengupta}
	\email{rb@iiserbpr.ac.in}
	\affiliation{Department of Mathematical Sciences, Indian
		Institute of Science Education and Research Berhampur, Laudigam,
		Konisi, Berhampur 760003,
		Odisha, India}
	\begin{abstract}
		A set of orthogonal product states is said to exhibit \emph{quantum nonlocality without entanglement} if it is locally indistinguishable, i.e. no sequence of local operations and classical communication (LOCC) can perfectly discriminate the states. Building on this foundational idea, recent studies have highlighted the phenomenon of \emph{genuine activation of hidden nonlocality}, where a set of initially distinguishable orthogonal states becomes locally indistinguishable through orthogonality-preserving LOCC transformations. In this letter, we establish that any complete orthogonal product basis that is initially locally distinguishable remains so under all orthogonality-preserving local projective measurements, thereby ruling out activation via orthogonality-preserving local projective measurements and classical communication. We further introduce and formalise the notions of 
        \emph{strongly local sets}, namely locally distinguishable sets that remain non-activable under all bipartitions. Interestingly, the study of \emph{local activability} of distinguishable sets is useful to characterise the boundary between LOCC distinguishability and its irreversible loss in multipartite systems. Our results provide a rigorous structural understanding of local-to-nonlocal transitions in quantum state discrimination.
	\end{abstract}
	
	\maketitle
    Local distinguishability of quantum states under Local Operations and Classical Communication (LOCC) \cite{BennettUPB1999,Walgate2000,Virmani,Ghosh2001,Groisman,Walgate2002,Divincinzo,Horodecki2003,Fan2004,Ghosh2004,Nathanson2005,Watrous2005,Niset2006,Ye2007,Fan2007,Runyo2007,somsubhro2009,Feng2009,Runyo2010,Yu2012,Yang2013,Zhang2014,somsubhro2010,yu2014,somsubhro2014,somsubhro2016,bennett1996,popescu2001,xin2008,somsubhro2009(1)} offers profound insights into the nonlocal features of quantum mechanics. A particularly striking phenomenon in this context is \emph{nonlocality without entanglement}~\cite{Zhang2015,Wang2015,Chen2015,Yang2015,Zhang2016,Xu2016(2),Zhang2016(1),Xu2016(1),Halder2019strong nonlocality,Halder2019peres set,Xzhang2017,Xu2017,Wang2017,Cohen2008,somsubhro2018,zhang2018,Halder2018,Yuan2020,Rout2019,bhunia2020,bhunia2023,biswas2023,Zhang2019,bhunia2022}, wherein a set of orthogonal quantum states--despite being unentangled--cannot be perfectly distinguished using LOCC alone. This departure from classical expectations has attracted significant attention, linking state distinguishability to fundamental aspects of quantum communication like data hiding \cite{terhal58,divincenzo580,lamidatahiding,terhaldatahiding,chaves2020,wehner2020,winterdatahiding,haydendatahiding} and secret sharing \cite{rahaman330,markham309,wang320}. This intriguing limitation naturally motivates a systematic investigation of the operational constraints imposed by LOCC on the extraction of classical information from composite quantum systems.

    A composite quantum system may exhibit intrinsically nonclassical properties. Classical information encoded in states of a composite quantum system involving spatially separated subsystems, may not always be decodable under LOCC. Such sets of states are called nonlocal due to their indistinguishable nature under LOCC. To elaborate a little, suppose a state is secretly chosen from a known set of states of a bipartite system shared between two distant parties, say, Alice and Bob. Their goal is to locally figure out the exact identity of the chosen state. Local quantum state discrimination process plays a prominent role in exploring the restrictions put forward by LOCC \cite{LOCC} on quantum systems with spatially separated subsystems.

	While entanglement is a key resource in quantum information, it is now well understood that even product states can exhibit strong nonlocal behavior when the number of parties is more or the structure of the state set is nontrivial. Such sets include Unextendible Product Bases (UPBs), which are locally indistinguishable despite consisting entirely of orthogonal product states~\cite{Divincinzo,BennettUPB1999}. This has motivated a broader investigation into the conditions under which local indistinguishability arises and whether it can be operationally induced through local protocols \cite{subrata2024,bandhyopadhyay201,Li2022}.
	
	A restricted yet operationally significant subclass of LOCC is the paradigm of \emph{Local Projective Measurements with Classical Communication} (LPCC) \cite{xin2008}. In this setting, parties are limited to performing local projective measurements and coordinating their actions through classical communication. Although strictly weaker than general LOCC, LPCC captures essential aspects of local information processing in realistic experimental conditions \cite{subrata2024,xin2008}. A key phenomenon of interest in this context is the \emph{activation of hidden nonlocality} \cite{bandhyopadhyay201}, where an initially locally distinguishable set of orthogonal states can evolve--under orthogonality-preserving local projective measurements--into a set that is no longer distinguishable by any sequence of LOCC operations.\footnote{Note that when we say a set is distinguishable or indistinguishable, the underlying class of operation is LOCC. But given a distinguishable set, we check its transformation to an indistinguishable one under LPCC here.} Such activation reveals hidden nonlocal features embedded within otherwise classically accessible state sets \cite{bandhyopadhyay201,subrata2024,Li2022,GhoshStrongActivation2022}. From an operational standpoint, it is natural to ask whether the full generality of LOCC is required to reveal local indistinguishability. Although LOCC protocols allow arbitrary local POVMs, implementing such measurements typically requires ancillary systems and additional unitary interactions, making them experimentally more demanding than standard projective measurements. Furthermore, in our case post-measurement states are particularly important. Now, in case of a POVM, it is not always possible to accurately determine the post-measurement state. Thus, we stick to local projective measurements only. These motivate the study of the restricted class of protocols consisting of LPCC, which are operationally simpler and closer to realistic implementations. Interestingly, while LPCC suffices in several scenarios-for example in the characterization of $2\otimes n$ state discrimination by Walgate {et al.}~\cite{Walgate2002}, there exist sets of orthogonal product states that are distinguishable by LOCC but not by LPCC, as shown by Bennett \emph{et al.}~\cite{BennettUPB1999}. Understanding the gap between LOCC and LPCC therefore provides deeper insight into the operational structure of local indistinguishability.
	
	In this work, we focus on the following general question: Given any complete orthonormal locally distinguishable basis, can it be transformed to a locally indistinguishable one by orthogonality-preserving LOCC? As described earlier, for our case, it is more important to stick to local projective measurements. Furthermore, since we want to begin with a distinguishable set, it must not contain any entangled state, otherwise, it will be a locally indistinguishable set \cite{Horodecki2003}. Clearly, now the aforesaid general question boils down to the following: given any complete orthonormal product basis which is distinguishable under LOCC, can it be transformed to a set which is indistinguishable under LOCC, via orthogonality-preserving local projective measurements and classical communication? This is depicted in Fig.~\ref{statement}. In other words, we are interested in the activation of quantum nonlocality without entanglement. To answer this question, it is important to understand the structure and behavior of \emph{complete orthogonal product bases} (COPB) under LPCC protocols. In this workwe shall consider orthogonality preserving LPCC (OP-LPCC). We prove that such sets are \emph{non-activable}, meaning that no orthogonality-preserving local projective measurement even with classical communication, can transform a COPB into a locally indistinguishable set. This establishes a strong constraint on the evolution of complete product bases and rules out the possibility of activating hidden nonlocality through OP-LPCC in these cases.
\vspace{-4mm}
    \begin{figure}[h!]
		\centering
		\includegraphics[width=0.51\textwidth, clip=true, trim = 10mm 20mm 30mm 20mm]{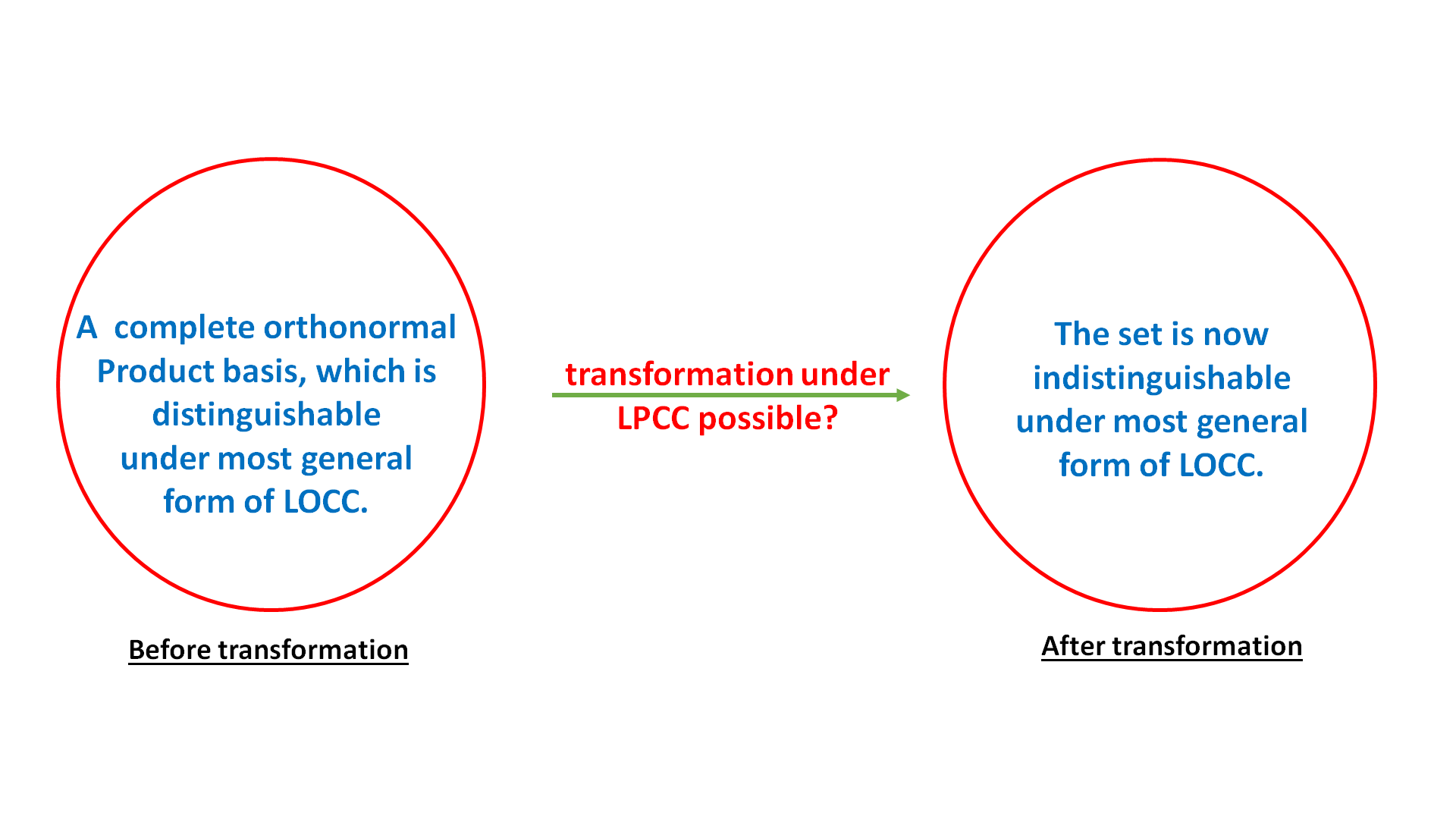}
		\caption{Diagram to illustrate the statement of the problem.}\label{statement}
	\end{figure}
\vspace{-4mm}    
\par The structural rigidity of COPB under LPCC not only rules out such activation in bipartite systems but also generalizes to multipartite scenarios. In these systems, one might expect that more intricate partitions and coordination strategies could induce activation \cite{subrata2024}. We show that,  Surprisingly, this may not be the case always. Complete orthogonal product bases in multipartite Hilbert spaces remain locally distinguishable and \emph{resist activation across all bipartitions}. We refer to this robust property as \emph{strong locality}. That is, no subset of parties--regardless of how they are grouped--can collaborate via orthogonality-preserving local projective measurements (OPLPM) to transform the set into one that is locally indistinguishable \cite{subrata2024,GhoshStrongActivation2022}. This reveals a profound form of non-activability, establishing the complete product bases as structurally stable configurations in the broader landscape of quantum nonlocality. Our results thus contribute to a deeper understanding of the limitations of LPCC \cite{xin2008}, the geometry of product state sets \cite{Li2022}, and the operational emergence of nonlocality in the absence of entanglement \cite{BennettUPB1999}. In this way, we arrive to the conclusion that the `incompleteness' for the given set of orthogonal product states constitutes a necessary condition for the activation of hidden nonlocality under OP-LPCC.
	
	In order to rigorously formulate and analyze the activation behavior of orthogonal product states under local constraints, we require a precise mathematical framework for local measurements and distinguishability \cite{Walgate2002}. We now recall the standard formalism of quantum measurements, followed by key definitions that will be used throughout our analysis.
	
	A measurement on a $d$-dimensional quantum system can be expressed as a set of positive operator-valued measure (POVM) elements $\left\{M_k\right\}_k$. These elements are  positive semidefinite Hermitian matrices that satisfy the completeness relation $\sum_k M_k=\mathrm{I}_d$, where $\mathrm{I}_d$ is the identity matrix of order $d$. We will first review some of the definitions used here.
\begin{defin}
		\cite{Walgate2002, Halder2018} If all the POVM elements of a measurement structure, corresponding to a discrimination task of a given set of states, are proportional to the identity matrix, then such a measurement is not useful to extract information for this task and is called a $trivial\;measurement$. On the other hand, if not all POVM elements of a measurement
		are proportional to the identity matrix then the measurement is said to be a \emph{nontrivial measurement}.
	\end{defin}
    
	\begin{defin}
		\cite{Walgate2002, Halder2018} Consider a local measurement to distinguish a fixed set of pairwise orthogonal quantum states. After performing that measurement, if the post measurement states are also	pairwise orthogonal to each other then such a measurement is said to be an $orthogonality-preserving\; local\;measurement$ (OPLM).
	\end{defin}
	\begin{defin}
		\cite{Halder2019strong nonlocality} A set of orthogonal quantum states is $locally\; irreducible$ if it is not possible to eliminate one or more quantum states from the set by nontrivial OPLM.
	\end{defin} 
	\begin{defin}
		A set of orthogonal quantum states is $locally\; indistinguishible$  if it is not possible to distinguish the whole set by nontrivial OPLM. However, for a locally indistinguishable set, it may or may not be possible to eliminate one or more states by OPLM. 
	\end{defin}
	Therefore, it is by definition implied that all locally irreducible states are locally indistinguishable but the converse is not always true.
	
	We now formally characterize the structure-preserving nature of OPLPM when applied to complete product bases. Specifically, we show that such measurements cannot generate any new directions or evolve the system beyond the original set.

	\begin{prop}
	 Let \( \mathcal{S} = \{ |\psi_i\rangle = |a_i\rangle \otimes |b_i\rangle \}_{i=1}^n \) be a set of mutually orthogonal product states in a finite-dimensional Hilbert space \( \mathcal{H}_A \otimes \mathcal{H}_B \) and \( \{P_k\} \) be a nontrivial OPLPM acting on any subsystem of \( \mathcal{S} \). If \( \mathcal{S} \) is a complete orthonormal product basis of \( \mathcal{H}_A \otimes \mathcal{H}_B \), then for each outcome \( k \), the nonzero post-measurement states form a proper subset of \( \mathcal{S} \).
		\label{thm:OPLM-completeness}
	\end{prop}
	\begin{proof}
		Suppose \( \mathcal{S} = \{ |\psi_i\rangle = |a_i\rangle \otimes |b_i\rangle \}_{i=1}^{d_A d_B} \) be a complete orthonormal product basis for \( \mathcal{H}_A \otimes \mathcal{H}_B \). Let Alice make a nontrivial OPLPM \( \{P_k\} \) on \( \mathcal{H}_A \). Then the post-measurement state for outcome \( k \) is of the form,
		\[
		|\psi_i^{(k)}\rangle = (P_k \otimes \mathbb{I}) |\psi_i\rangle = P_k |a_i\rangle \otimes |b_i\rangle.
		\]
		
		As, \( P_k \) is an orthogonal projector,
		\[
		P_k |a_i\rangle =
		\begin{cases}
			|a_i\rangle & \text{if } |a_i\rangle \in \operatorname{supp}(P_k), \\
			0 & \text{if } |a_i\rangle \in \overline{\operatorname{supp}(P_k)}.
		\end{cases}
		\]
		Thus,
		\[
		|\psi_i^{(k)}\rangle =
		\begin{cases}
			|\psi_i\rangle & \text{if } |a_i\rangle \in \operatorname{supp}(P_k), \\
			0 & \text{if } |a_i\rangle \in \overline{\operatorname{supp}(P_k)}.
		\end{cases}
		\]
	Here,	\(\operatorname{supp}   (P_k)\) denotes the support of \(P_k\), while 
      \(\overline{\operatorname{supp}(P_k)}\) denotes its orthogonal complement. Now, for the vectors \( |a_i\rangle \in \mathcal{H}_A \) which are not fully contained in \( \operatorname{supp}(P_k) \), the action of the projector \( P_k \) must gives post-measurement states, which are not proportional to any member of \( \mathcal{S} \). \\
      \textbf{Case (a)}- Firstly, let us assume that there exists exactly one state $|\psi_i\rangle\in\mathcal{S}$, for which, the resulting state \( |\psi_i^{(k)}\rangle = (P_k \otimes \mathbb{I})|\psi_i\rangle \) becomes a new product state that is not proportional to any member of \( \mathcal{S} \). Since the set \( \mathcal{S} \) is a complete orthonormal basis of the full space \( \mathcal{H}_A \otimes \mathcal{H}_B \), the state \(|\psi_i^{(k)}\rangle = (P_k \otimes \mathbb{I})|\psi_i\rangle \) must lie in the span of \{\( |\psi_i\rangle \)\}. So \(|\psi_i^{(k)}\rangle \; \text{and} \; |\psi_i\rangle \) are proportional.
      
    \textbf{Case (b)}- Let us assume that
    for a fixed $i$ and $j,\;
    \, i \neq j,$ two states 
    $|\psi_i\rangle, |\psi_j\rangle \in \mathcal{S}$, such that the post-measurement states 
    \[
    |\psi_i^{(k)}\rangle = (P_k \otimes \mathbb{I})|\psi_i\rangle, 
    \qquad 
    |\psi_j^{(k)}\rangle = (P_k \otimes \mathbb{I})|\psi_j\rangle
    \]
   are new product states not proportional to any member of $\mathcal{S}$. Since the measurement acts only on Alice's system, both 
   $|\psi_i^{(k)}\rangle$ and $|\psi_j^{(k)}\rangle$ must lie in the span of the original pair $\{|\psi_i\rangle, |\psi_j\rangle\}$.
   
    (i) If Bob's local parts 
    $|b_i\rangle, |b_j\rangle$ are orthogonal, then the post-measurement states must reduce to 
    $|\psi_i\rangle$ and $|\psi_j\rangle$ by the previous case. Hence, no new product state can arise.  

    (ii) If Bob's system is fixed, i.e., 
    $|b_i\rangle = |b_j\rangle$, then the action of $P_k$ on Alice's 
    states $|a_i\rangle, |a_j\rangle$ generally produces two nonorthogonal states. This contradicts the orthogonality-preserving 
    condition, which requires that
    \[
    \langle \psi_i^{(k)} | \psi_j^{(k)} \rangle 
    = \langle P_k a_i | P_k a_j \rangle \cdot 
      \langle b_i | b_j \rangle = 0,
    \qquad \forall \, i \ne j .
    \]

  Thus, in either case, assuming the existence of two new product 
  states leads to a contradiction. Therefore, no two such new product states can exist simultaneously under an OPLPM. Hence, nonzero post-measurement states must be contained in \( \mathcal{S} \).
\end{proof}
\begin{figure}[h!]
		\centering	\includegraphics[width=0.37\textwidth]{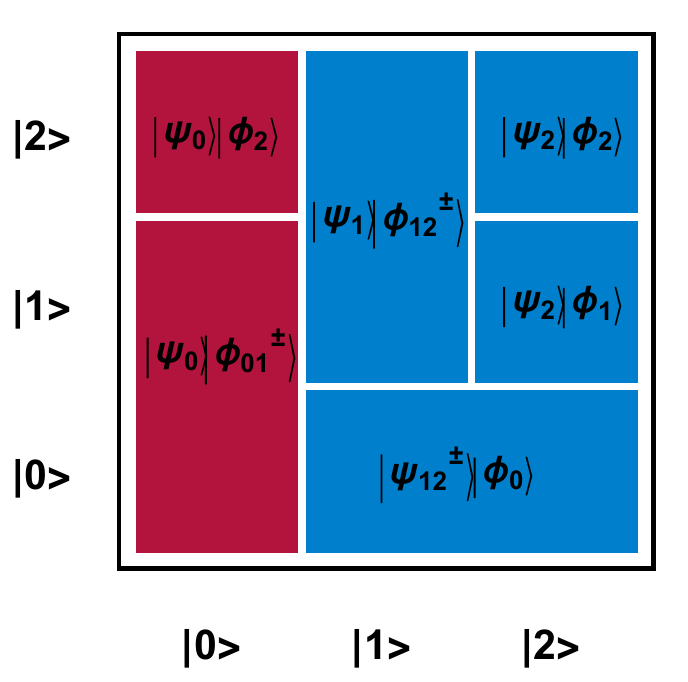}
		\caption{Tiling diagram for the states in \( \mathbf{\mathcal{S}_1} \). The red colored region indicates the support of Alice's measurement outcome \( \mathcal{K}_1^A \) corresponding to the measurement set up $\mathcal{K}_A \equiv\left\{\mathcal{K}_1^A:=\right.$ $P\left[|\mathbf{0}\rangle_A\right],$ $\mathcal{K}_2^A := \mathbb{I} - \mathcal{K}_1^A\}$, resulting in post-measurement states which are all contained in the original set. The horizontal side of the tiles belongs to (A)lice and the vertical side of the tiles belongs to (B)ob.}\label{fig44-1}
	\end{figure}
  We give an example for better understanding. Consider the set $\mathbf{\mathcal{S}_1} =\left\{\left|\phi_i\right\rangle_{A B}\right\}\in \mathbf{\mathbf{C}^3 \otimes \mathbf{C}^3}$ \cite{bhunia2023}, by, 
	\begin{multline}
		\mathbf{\mathcal{S}_1}=
		\begin{Bmatrix}
			{|\Psi_0\rangle_{A}|{\Phi}_{01}^{\pm}\rangle_{B},}\;
			{|\Psi_0\rangle_{A}|{\Phi}_{2}\rangle_{B},}\;
			{|{\Psi}_{12}^{\pm}\rangle_{A}|\Phi_0\rangle_{B},}\;\\
			{|\Psi_1\rangle_{A}|{\Phi}_{12}^{\pm}\rangle_{B},}\;
			{|{\Psi}_{2}\rangle_{A}|\Phi_1\rangle_{B},}\;
			{|{\Psi}_{2}\rangle_{A}|{\Phi}_{2}\rangle_{B}}\\
		\end{Bmatrix}
		\label{1}
	\end{multline}
	where,	${|{\Psi}_{ij}^{\pm}\rangle_{A}}=\left(\frac{|\mathbf{i}\rangle\pm |\mathbf{j}\rangle}{\sqrt{2}}\right)_{A},\;{|{\Phi}_{ij}^{\pm}\rangle_{B}}=\left(\frac{|\mathbf{i}\rangle\pm |\mathbf{j}\rangle}{\sqrt{2}}\right)_{B},$ and ${|{\Psi}_{k}\rangle_{A}}=|\mathbf{k}\rangle_{A},\;{|{\Phi}_{k}\rangle_{B}}=|\mathbf{k}\rangle_{B}$, see Fig.~\ref{fig44-1}.
	
    Each state in $\mathbf{\mathcal{S}_1}$ is a product of either a standard basis vector \( |i\rangle \) or a superposition \( \tfrac{1}{\sqrt{2}}(|i\rangle \pm |j\rangle) \), i.e., of the form $|\phi_i\rangle = |a_i\rangle \otimes |b_i\rangle,$
	and the cardinality of $\mathbf{\mathcal{S}_1}$ matches $\dim \mathbf{(\mathbf{C}^3 \otimes \mathbf{C}^3)}$. Since, \(\{\ket{\phi_i}\}\) are pairwise orthogonal, the states of \( \mathcal{S}_1 \) form a complete orthonormal product basis.
	
	Now, consider any OPLPM applied by any party, say Alice. According to the Proposition \ref{thm:OPLM-completeness}, the post-measurement states under such a Projection-valued measurement (PVM) must be either zero or remain within \( \mathbf{\mathcal{S}_1} \). That is, no new product vector outside this set can emerge after the measurement. Thus, \( \mathbf{\mathcal{S}_1} \) serves as an explicit example where OPLPM cannot induce evolution within the product structure. It demonstrates the fact that complete orthonormal product bases are invariant (up to elimination of states) under OPLPM.
	\begin{cor} If \( \mathcal{S} \) is not complete, then there may exist a nontrivial OPLPM \( \{P_k\} \), such that for at least one outcome \( k \), the nonzero post-measurement states do not form a proper subset of \( \mathcal{S} \).
	\end{cor}
	Let us now present a concrete example to demonstrate the case of an incomplete basis, where an OPLPM can evolve product states into new directions not originally present in the given set. Consider, the set $\mathbf{\mathcal{S}_2} =\left\{\left|\phi_i\right\rangle_{A B}\right\}_{i=1}^{10}\left(\in \mathbf{\mathbf{C}^3 \otimes \mathbf{C}^6}\right)$ \cite{GhoshStrongActivation2022}, where
	\begin{multline}
		\begin{aligned}
			& \left|\phi_1\right\rangle_{A B}=|\mathbf{0}\rangle_A|\mathbf{0}-\mathbf{1}+\mathbf{4}-\mathbf{5}\rangle_B \\
			& \left|\phi_2\right\rangle_{A B}=|\mathbf{2}\rangle_A|\mathbf{1}-\mathbf{2}+\mathbf{5}-\mathbf{3}\rangle_B \\
			& \left|\phi_3\right\rangle_{A B}=|\mathbf{1-2}\rangle_A|\mathbf{0}-\mathbf{4}\rangle_B \\
			& \left|\phi_4\right\rangle_{A B}=|\mathbf{0-1}\rangle_A|\mathbf{2}-\mathbf{3}\rangle_B \\
			& \left|\phi_5\right\rangle_{A B}=|\mathbf{0+1+2}\rangle_A|\mathbf{0}+\mathbf{1}+\mathbf{2}+\mathbf{3}+\mathbf{4}+\mathbf{5}\rangle_B
		\end{aligned}\label{5}
	\end{multline}
	
	The set \( \mathbf{\mathcal{S}_2} \) (see Fig.~\ref{f44-3}) consists of orthogonal product states with cardinality strictly less than \( \dim(\mathbf{\mathbf{C}^3 \otimes \mathbf{C}^6})\), and hence forms an \emph{incomplete} orthogonal product set in \( \mathbf{\mathbf{C}^3 \otimes \mathbf{C}^6} \). Now, let Bob perform a local measurement $\mathcal{K}_B \equiv\left\{\mathcal{K}^B_1:=P\left[(|\mathbf{0}\rangle,|\mathbf{1}\rangle,|\mathbf{2}\rangle)_B\right], \mathcal{K}^B_2:=\right.$ $\left.P\left[(|\mathbf{3}\rangle,|\mathbf{4}\rangle,|\mathbf{5}\rangle)_B\right]\right\}$. If $\mathcal{K}^B_1$ clicks they end up with,
	$$
	\left\{\begin{array}{c}
		|\mathbf{0}\rangle_A|\mathbf{0}-\mathbf{1}\rangle_B,|\mathbf{2}\rangle_A|\mathbf{1}-\mathbf{2}\rangle_B, 
		|\mathbf{1-2}\rangle_A|\mathbf{0}\rangle_B,|\mathbf{0-1}\rangle_A|\mathbf{2}\rangle_B, \\
		|\mathbf{0+1+2}\rangle_A|\mathbf{0}+\mathbf{1}+\mathbf{2}\rangle_B
	\end{array}\right\}
	$$
	Clearly, although all states in \( \mathbf{\mathcal{S}_2^{(1)}} \)(set of post-measurement reduced states for $\mathcal{K}^B_1$) are orthogonal and product, some of them are not proportional to any of the original states in \( \mathbf{\mathcal{S}_2} \), thereby implying $\mathbf{\mathcal{S}_2^{(1)} \nsubseteq \mathcal{S}_2}.$
	\begin{figure}[h!]
		\centering
		\includegraphics[width=0.48\textwidth]{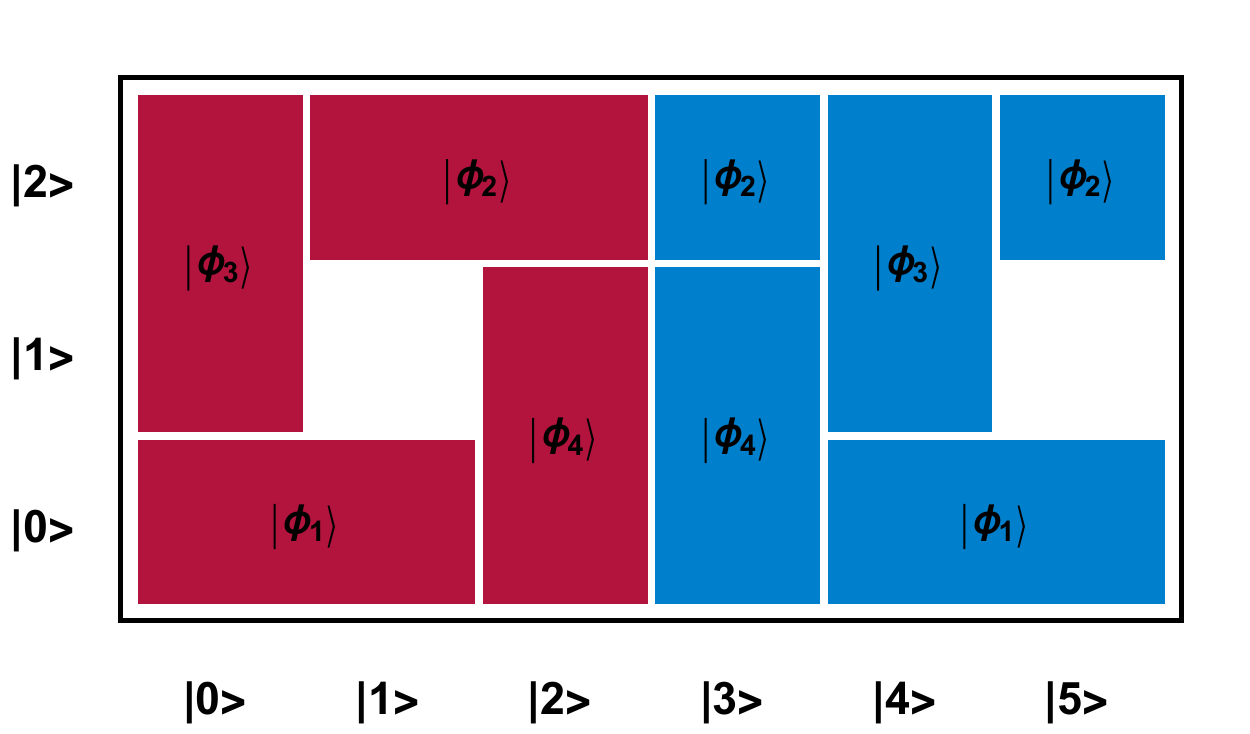}
		\caption{Tiling diagram for the states in \( \mathcal{S}_2 \). The outlined region indicates the support of Bob's measurement outcome \( \mathcal{K}_1^B \), resulting in post-measurement states not all of which are contained in the original set. The horizontal side of the tiles belongs to (B)ob and the vertical side of the tiles belongs to (A)lice.}
		\label{f44-3}
	\end{figure}
	Hence, the post-measurement states do not lie within the set \( \mathbf{\mathcal{S}_2} \), and new product vectors emerge through the action of an OPLPM. This explicitly confirms that when the initial set \( \mathbf{\mathcal{S}}\) is incomplete, local OPLPM can indeed cause the system to evolve into new orthogonal product directions.
	\par This contrasting behavior highlights a critical distinction: the emergence of new product vectors under local PVM is intimately tied to whether the original product basis spans a product subspace. In particular, if the initial set \( \mathcal{S} \) forms a complete orthonormal product basis of a tensor product subspace, the dynamics under OPLPM remain confined within \( \mathcal{S} \), as we formalize in the following proposition.
	
	\begin{prop}
		Let \( \mathcal{S} = \{ |\psi_i\rangle = |a_i\rangle \otimes |b_i\rangle \}_{i=1}^n \) be an incomplete orthogonal product basis in a finite-dimensional Hilbert space \( \mathcal{H}_A \otimes \mathcal{H}_B \) and \( \{P_k\} \) be a nontrivial OPLPM acting on any subsystem of \( \mathcal{S} \). If \( \mathcal{S} \) forms a complete orthonormal product basis of a subspace \( \mathcal{H}_{A_{\mathcal{S}}} \otimes \mathcal{H}_{B_{\mathcal{S}}} \subseteq \mathcal{H}_A \otimes \mathcal{H}_B \), then for each outcome \( k \), the nonzero post-measurement states form a proper subset of \( \mathcal{S} \).
	\end{prop}
	
	\begin{proof}
		Since the projective measurement \( \{P_k\} \) is orthogonality-preserving on \( \mathcal{S} \), the post-measurement vectors \( (P_k \otimes \mathbb{I})|\psi_i\rangle \) (Alice's measurement) remain mutually orthogonal for all \( i \). Given that \( \mathcal{S} \) is a complete orthonormal product basis of the subspace \( \mathcal{H}_{A_{\mathcal{S}}} \otimes \mathcal{H}_{B_{\mathcal{S}}} \), any linear combination of its elements must remain within that subspace.
		
		The measurement operator (\( P_k \otimes \mathbb{I} \)) acts only on \( \mathcal{H}_{A_{\mathcal{S}}} \), and thus maps \( \mathcal{H}_{A_{\mathcal{S}}} \otimes \mathcal{H}_{B_{\mathcal{S}}} \) into itself. Therefore, the post-measurement states \( (P_k \otimes \mathbb{I})|\psi_i\rangle \in \mathcal{H}_{A_{\mathcal{S}}} \otimes \mathcal{H}_{B_{\mathcal{S}}} \) and are still product states, being projections of product states.
		
		Since \( \mathcal{S} \) spans the full product basis of that subspace, no new orthogonal product states exist outside \( \mathcal{S} \). Therefore, by Proposition \ref{thm:OPLM-completeness}, the nonzero post-measurement vectors must belong to \( \mathcal{S} \). Hence, no new product vectors are introduced.
	\end{proof}
	
	\begin{rem}
		The above proposition emphasizes that when a product basis is complete within a given subspace, OPLPM cannot induce any extension or evolution beyond that basis. This is crucial for understanding when activation effects can or cannot arise.
	\end{rem}
	
	This distinction between subspace--complete vs. structurally--incomplete sets plays a pivotal role in determining the dynamical behavior under local measurements. Specifically, the absence of a tensor-product subspace structure enables the generation of novel product vectors, offering a clear pathway for activation via local PVM.

	\begin{cor}
        Let \( \mathcal{S} = \{ |\psi_i\rangle = |a_i\rangle \otimes |b_i\rangle \}_{i=1}^n \) be an incomplete orthogonal product basis in a finite-dimensional Hilbert space \( \mathcal{H}_A \otimes \mathcal{H}_B \) and \( \{P_k\} \) be a nontrivial OPLPM acting on any subsystem of \( \mathcal{S} \). If \( \mathcal{S} \) does not form a complete orthonormal product basis of any subspace \( \mathcal{H}_{A_{\mathcal{S}}} \otimes \mathcal{H}_{B_{\mathcal{S}}} \subseteq \mathcal{H}_A \otimes \mathcal{H}_B \), then there may exist at least one outcome \( k \), such that the nonzero post-measurement states do not form a proper subset of \( \mathcal{S} \).
	\end{cor}

	\begin{proof}
		Let \( \mathcal{S} = \{ |\psi_i\rangle = |a_i\rangle \otimes |b_i\rangle \}_{i=1}^n \in \mathcal{H}_A \otimes \mathcal{H}_B \) be an incomplete orthogonal product basis that does \emph{not} form a complete orthonormal product basis of any subspace \( \mathcal{H}_{A_{\mathcal{S}}} \otimes \mathcal{H}_{B_{\mathcal{S}}} \subseteq \mathcal{H}_A \otimes \mathcal{H}_B \). Then, at least one of \( \{ |a_i\rangle \} \) and \( \{ |b_i\rangle \} \) does not form a full orthonormal basis for any subsystem, say, it is subsystem $A$.
		
		Now, consider an orthogonality-preserving local projective measurement \( \{P_k\} \) on subsystem \( A \). For each measurement outcome \( k \), the post-measurement states are
		\(
		|\phi_i^{(k)}\rangle = (P_k \otimes \mathbb{I}) |\psi_i\rangle.
		\)
		Since \( P_k \) acts only on \( A \), the new local component on Alice's side becomes \( P_k |a_i\rangle \), which may lie \emph{outside} the span of \( \{ |a_i\rangle \} \). Because the original \( \{ |a_i\rangle \} \) do not form a full orthonormal basis of any subspace, the resulting vectors \( P_k |a_i\rangle \otimes |b_i\rangle \) may no longer lie within \( \mathrm{span}(\mathcal{S}) \), and new product states may emerge.
	\end{proof}
	
	\begin{rem}
		The key structural distinction lies in whether the original product basis \( \mathcal{S} \) forms a complete basis for a subspace. If it does, the evolution remains closed within \( \mathcal{S} \) under OPLPM. Otherwise, such evolution may extend beyond \( \mathrm{span}(\mathcal{S}) \), generating new product vectors.
	\end{rem}

	This behavior illustrates a key difference between complete and incomplete sets: the lack of completeness allows new orthogonal product directions to arise under local measurements. To capture this transformation formally, we introduce the notion of local activability.
	\begin{defin}
		A locally distinguishable set $\mathcal{S}$ of multipartite orthogonal states is said to be locally activable if it can be transformed to a set of locally indistinguishable orthogonal states via OPLM.
	\end{defin}
	\begin{theorem}	\label{thm:LPCC-non-activation}
		Let \( \mathbf{\mathcal{S}} \) be a complete orthogonal product basis of \( \mathcal{H}_A \otimes \mathcal{H}_B \) which is initially distinguishable by LOCC. Then \( \mathbf{\mathcal{S}} \) cannot be made locally indistinguishable via any LPCC protocol. In other words, \( \mathcal{S} \) is not activable under LPCC.
	\end{theorem}
	\begin{proof}
		The result follows directly from Proposition \ref{thm:OPLM-completeness}. Suppose \( \mathbf{\mathcal{S}} \) is a complete orthogonal product basis which is initially distinguishable. Any local PVM on one subsystem, say, Alice's side will map each state in \( \mathbf{\mathcal{S}} \) to either itself or zero, as shown in Proposition \ref{thm:OPLM-completeness}. Thus, the post-measurement set is a subset of \( \mathbf{\mathcal{S}} \), containing only those states that survive the projection.
		
		Now, if possible let us assume that \( \mathbf{\mathcal{S}} \) is activable by LPCC. Then there exists a sequence of local PVM with classical communication such that, after some step, the remaining post-measurement states become locally indistinguishable. But this implies that some subset of \( \mathbf{\mathcal{S}} \) is locally indistinguishable.
		
		However, this contradicts the assumption that \( \mathbf{\mathcal{S}} \) is initially distinguishable and no subset of a locally distinguishable set of orthogonal states can become locally indistinguishable.
	\end{proof}   
    See Appendix \ref{appendixA} for the complete description. In contrast, the set \( \mathcal{S}_2 \) (see Fig.~\ref{f44-3}), defined on the same Hilbert space, is also initially distinguishable but incomplete. According to Theorem \ref{thm:LPCC-non-activation}, such a set may become activable. 
	
	The non-activability result for complete orthogonal product bases under LPCC in bipartite systems naturally extends to multipartite settings. In a multipartite scenario, OPLPM still acts independently on each subsystem without introducing entanglement or new product directions. Consequently, the structural constraints imposed by completeness and product form continue to prevent the emergence of local indistinguishability through LPCC. We now formalize this extension in the following theorem.
	\begin{theorem}
		\label{thm:multipartite-non-activation}
		Let \( \mathcal{S} \) be a complete orthogonal product basis in the multipartite Hilbert space \( \mathcal{H} = \mathcal{H}_1 \otimes \mathcal{H}_2 \otimes \cdots \otimes \mathcal{H}_n \), where \( n \geq 2 \). If \( \mathcal{S} \) is initially distinguishable by LOCC, then \( \mathcal{S} \) cannot be made locally indistinguishable via any LPCC protocol. That is, \( \mathcal{S} \) is not activable under LPCC.
	\end{theorem}
	
	\begin{proof} The proof is a natural extension of Theorem \ref{thm:LPCC-non-activation} to multipartite systems. Suppose \( \mathcal{S} \) is a complete orthogonal product basis of the multipartite space \( \mathcal{H}_1 \otimes \cdots \otimes \mathcal{H}_n \), and assume that it is initially distinguishable. In any LPCC protocol, local PVM are performed on one subsystem at a time. For a complete product basis, each local PVM either eliminates a state  or preserves it, without altering the product structure or introducing new directions. Thus, after each round of local measurement and classical communication, the post-measurement states remain a subset of the original set \( \mathcal{S} \). Now, assume for contradiction that \( \mathcal{S} \) is activable. Then there exists a sequence of local PVM and classical communication that transforms \( \mathcal{S} \) into a locally indistinguishable subset. But this would imply that some subset of \( \mathcal{S} \) is locally indistinguishable, contradicting the assumption that \( \mathcal{S} \) is initially distinguishable. \end{proof}
	
	While Theorem \ref{thm:multipartite-non-activation} addresses non-activability in the full multipartite setting, it is also important to consider the behavior of such sets under bipartitions of the system. This motivates the following definition, which captures a stronger form of local invariance.
	
	\begin{defin}
		\label{definationstronglocal}
		Consider a composite quantum system \( \mathcal{H} = \bigotimes_{i=1}^n \mathcal{H}_i \), with \( n \geq 3 \) and \( \dim \mathcal{H}_i \geq 3 \) for all \( i = 1, \dots, n \). A set of orthogonal product states \( \{ |\psi_i\rangle = |\alpha_i\rangle_1 \otimes |\beta_i\rangle_2 \otimes \cdots \otimes |\gamma_i\rangle_n \} \in \mathcal{H} \) is said to be \emph{strongly local} if it is not activable under LOCC in any possible bipartition of the parties.
	\end{defin} 
	
	The above definition formalizes the notion of strong local behavior by requiring local distinguishability to be preserved across all bipartitions. We now show that every initially distinguishable complete orthogonal product basis satisfies this property.
	
	\begin{theorem}
		\label{thm:strongly-local}
		Let \( \mathcal{S} \in \mathcal{H}_1 \otimes \mathcal{H}_2 \otimes \cdots \otimes \mathcal{H}_n \) be a complete orthogonal product basis which is initially distinguishable by LOCC. Then \( \mathcal{S} \) is strongly local. In other words, \( \mathcal{S} \) is not activable under LPCC in any bipartition of the subsystems.
	\end{theorem}
	
	\begin{proof}
		Let \( \mathcal{S} \) be a complete orthogonal product basis in the multipartite space \( \mathcal{H}_1 \otimes \cdots \otimes \mathcal{H}_n \), and suppose that it is initially distinguishable.
		
		Consider any bipartition of the subsystems, say \( A | B \), where \( A \) and \( B \) are disjoint subsets of \( \{1, 2, \ldots, n\} \) and \( A \cup B = \{1, 2, \ldots, n\} \). Under this bipartition, \( \mathcal{S} \) remains a set of product states in \( \mathcal{H}_A \otimes \mathcal{H}_B \). Since \( \mathcal{S} \) is complete in the full space, it remains complete in any bipartition, and orthogonality is preserved.
		
		By Theorem \ref{thm:LPCC-non-activation}, a complete orthogonal product basis which is initially distinguishable cannot be made locally indistinguishable via LPCC in any bipartition. Therefore, \( \mathcal{S} \) is not activable under any such bipartition.
		
		Since the bipartition was arbitrary, it follows that \( \mathcal{S} \) is not activable in any possible bipartition. Hence, by Definition~\ref{definationstronglocal}, \( \mathcal{S} \) is strongly local.
	\end{proof} 
	The above findings reveal a fundamental structural constraint on the activation of hidden nonlocality. Specifically, the completeness of an orthogonal product basis inherently restricts any LPCC protocol from transforming the set into a locally indistinguishable one. Thus, only incomplete product bases can serve as potential candidates for such activation. This positions COPB as extreme, strongly local structures, offering the least utility from the application perspective of \emph{nonlocality without entanglement}.

	\textit{Discussion.-} In this work, we have systematically explored the interplay between local PVM and the distinguishability of orthogonal product states in both bipartite and multipartite quantum systems. Our main result establishes that any COPB that is initially distinguishable remains so under all OPLPM, and hence cannot be activated into a locally indistinguishable set via LPCC. This non-activability is shown to extend robustly to all bipartitions in multipartite systems, leading to the notion of \emph{strongly local} sets-COPB that are immune to activation under any partitioning of subsystems.\\
	In contrast, we demonstrate that incomplete orthogonal product sets can exhibit qualitatively different behavior. Such sets may evolve under OPLPM into new orthogonal states not contained in the original span, allowing for the emergence of local indistinguishability where none previously existed. This leads us to formally define the concept of \emph{local activability}, characterizing those initially distinguishable sets that can be transformed into irreducible or indistinguishable configurations by local operations.\\    
	Our findings underscore a fundamental structural distinction: completeness in product bases imposes a rigid constraint on local evolution, while incompleteness leaves room for activation phenomena. These insights deepen the understanding of the boundary between local and nonlocal behavior in quantum state discrimination, and open pathways for further investigation into catalytic and entanglement-assisted transformations, as well as applications to quantum data hiding and nonlocality without entanglement.\\

	\noindent{\emph{Acknowledgments} A. Bhunia, acknowledges the support of ANRF-NPDF (File no. PDF/2025/002762).

\appendix*
	\section{The set \( \mathcal{S}_1 \) is initially distinguishable and not activable under LPCC}
	\label{appendixA}
	Now, consider the set $\mathcal{S}_1 =\left\{\left|\phi_i\right\rangle_{A B}\right\}\in \mathbf{\mathbf{C}^3 \otimes \mathbf{C}^3}$ \cite{bhunia2023}, by, 
	\begin{multline}
	\mathbf{\mathcal{S}_1}=
	\begin{Bmatrix}
		{|\Psi_0\rangle_{A}|{\Phi}_{01}^{\pm}\rangle_{B},}\;
		{|\Psi_0\rangle_{A}|{\Phi}_{2}\rangle_{B},}\;
		{|{\Psi}_{12}^{\pm}\rangle_{A}|\Phi_0\rangle_{B},}\;\\
		{|\Psi_1\rangle_{A}|{\Phi}_{12}^{\pm}\rangle_{B},}\;
		{|{\Psi}_{2}\rangle_{A}|\Phi_1\rangle_{B},}\;
		{|{\Psi}_{2}\rangle_{A}|{\Phi}_{2}\rangle_{B}}\\
	\end{Bmatrix}
	\label{A1}
\end{multline}
	where,	${|{\Psi}_{ij}^{\pm}\rangle_{A}}=\left(\frac{|\mathbf{i}\rangle\pm |\mathbf{j}\rangle}{\sqrt{2}}\right)_{A},\;{|{\Phi}_{ij}^{\pm}\rangle_{B}}=\left(\frac{|\mathbf{i}\rangle\pm |\mathbf{j}\rangle}{\sqrt{2}}\right)_{B},$ and ${|{\Psi}_{k}\rangle_{A}}=|\mathbf{k}\rangle_{A},\;{|{\Phi}_{k}\rangle_{B}}=|\mathbf{k}\rangle_{B}$, see Fig.~\ref{fig44-5}.
	\begin{figure}[h!]
		\centering
		\includegraphics[width=0.38\textwidth]{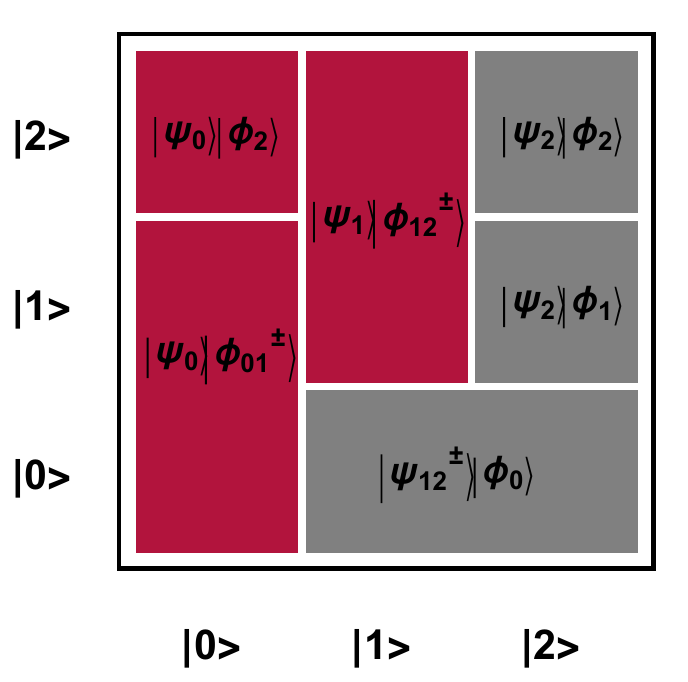}
		\caption{Representation of product states in ${\mathbb{C}}^{3}\otimes{\mathbb{C}}^{3}$. Tile indices correspond to consecutively ordered basis states of set $\mathcal{S}_{1}$, while tile colors indicate compatible measurement setups for both parties.}\label{fig44-5}
	\end{figure}
	\begin{theorem}
		The set $\mathcal{S}_1$ does not possess any activable nonlocality under LPCC.
	\end{theorem}
	\begin{proof} 
		Suppose Alice goes first, and let 
		${\mathbb{M}_{A}^m}^\dagger\mathbb{M}_{A}^m=[m^a_{ij}]_{3\times3}$ denote any arbitrary POVM operator of Alice with outcome
		$m$ such that the postmeasurement states $\quad\left\{\mathbb{M}_{A}^m \otimes I_{B} \left|\Psi_{i}\right\rangle,\right.$
		$|\Psi_{i}\rangle\in \mathcal{S}_1\}$ should be mutually orthogonal. Because $m^a_{i j}=0$ is necessary and sufficient for $m^a_{j i}=0, i<j$, we will only show $m^a_{i j}=0, i<j$, in the following. Then, considering the states $	|\Psi_0\rangle_{A}|{\Phi}_{01}^{+}\rangle_{B}$ and $ |\Psi_1\rangle_{A}|{\Phi}_{12}^{+}\rangle_{B},$ we know $\left\langle 0\left|\mathbb{M}_{A}^m\right| 1\right\rangle_{A} \left\langle 0+1|1+2\right\rangle_{B}=0.$ Thus, $m^a_{01}=m^a_{10}=0 .$ In the same way, for the states $	
		|\Psi_0\rangle_{A}|{\Phi}_{2}\rangle_{B},$ and $
		|\Psi_2\rangle_{A}|{\Phi}_{2}\rangle_{B},$
		we can compute $m^a_{02}=m^a_{20}=0$. Similarly if we choose the states $|\Psi_1\rangle_{A}|{\Phi}_{12}^{+}\rangle_{B},$ and $
		|{\Psi}_{2}\rangle_{A}|\Phi_1\rangle_{B}\}$,
		we can show $m^a_{12}=m^a_{21}=0$. Therefore, ${\mathbb{M}_{A}^m}^\dagger\mathbb{M}_{A}^m$ is diagonal and ${\mathbb{M}_{A}^m}^\dagger\mathbb{M}_{A}^m=\operatorname{diag}\left(\lambda_{0}, \lambda_{1},\lambda_{2}\right)$.
		Now considering $|{\Psi}_{12}^{\pm}\rangle_{A}|\Phi_0\rangle_{A},$ we get $\quad\left\langle1+2\left|\mathbb{M}_{A}^m\right|1-2\right\rangle_{A}\langle 0|0\rangle_{B}=0, \quad$ i.e.,
		$\left\langle1\left|\mathbb{M}_{A}^m\right|1\right\rangle-\left\langle 2\left|\mathbb{M}_{A}^m\right|2\right\rangle=0$. Therefore, ${\mathbb{M}_{A}^m}^\dagger{\mathbb{M}_{A}^m}$
		$=\operatorname{diag}\left(\lambda_{0}, \lambda, \lambda\right) .$
		If possible let us assume that $\lambda_{0}\neq0$ and $\lambda\neq0$. Then after Alice's measurement, Bob should do a nontrivial operation on his own system according to Alice's result. We denote $\mathbb{M}_{B}^m$ as Bob's operator. As we discussed above, by choosing suitable pair of states we can conclude that all the off-diagonal element of ${\mathbb{M}_{B}^m}^\dagger\mathbb{M}_{B}^m$ is equal to 0. Similarly for the diagonal element as we have discussed above, if we take $	|\Psi_0\rangle_{A}|{\Phi}_{01}^{\pm}\rangle_{B}\; \text{and}\; |\Psi_1\rangle_{A}|{\Phi}_{12}^{\pm}\rangle_{B},$ we finally get $m^b_{00}=m^b_{11}=m^b_{22}.$ Therefore ${\mathbb{M}_{B}^m}^\dagger\mathbb{M}_{B}^m$ is propotional to the identity operator, i.e., ${\mathbb{M}_{B}^m}^\dagger\mathbb{M}_{B}^m=\alpha_{0}I$, which is trivial operator and this contradicts our assumption that the set $\mathcal{S}_1$ is initially distinguishable. So, either $\lambda_{0}=0$ or $\lambda=0 .$
		Notice that this result also suggests that these states cannot be distinguished if Bob goes first. Now it is clear that if Alice goes first with a diagonal operator i.e., $\lambda_{0}=\lambda=1$, then the above set of states cannot be distinguished. So, Alice has to do non-trivial measurement first and this only happens when any one of $\lambda_{0}$, $\lambda$ not equal to zero. For that Alice only has two outcome measurement operators: $\quad \mathbb{M}_{A}^1=\operatorname{diag}(1,0,0)$ and $\mathbb{M}_{A}^2$ = $I-\mathbb{M}_{A}^1$ $=\operatorname{diag}(0,1,1)$, see Fig.~\ref{fig44-5}. If $\mathbb{M}_{A}^1$ clicks, Bob is able to distinguish the left states by projecting onto $\left|0\pm1\right\rangle$ and $\left|2\right\rangle$. If $\mathbb{M}_{A}^2$ clicks, it isolates the remaining states. It is then Bob's turn to do measurement. Following the method we used above, we can similarly prove that Bob's measurement must be $\mathbb{M}_{B}^1=\operatorname{diag}(1,0,0)$ and $\mathbb{M}_{B}^2$ $=\operatorname{diag}(0,1,1)$. The process will repeat a finite number of times and for each measurement outcomes for both parties the set $\mathcal{S}_1$ transforms only to a distinguishable set. This implies the fact that if the set is distinguishable (local) then for all possible nontrivial measurements, it is impossible to transform the set into an indistinguishable one. In other words, the set $\mathcal{S}_1$ is not activable through LPCC. Hence we complete the proof.
	\end{proof}
	\section{The set \( \mathcal{S}_2 \) is initially distinguishable and activable under LPCC}
	\label{appendixB}
	Consider the set $\mathcal{S}_2 =\left\{\left|\phi_i\right\rangle_{A B}\right\}_{i=1}^{5}\in \mathbf{\mathbf{C}^3 \otimes \mathbf{C}^6}$ \cite{GhoshStrongActivation2022}, where,
	\begin{multline}
		\begin{aligned}
			& \left|\phi_1\right\rangle_{A B}=|\mathbf{0}\rangle_A|\mathbf{0}-\mathbf{1}+\mathbf{4}-\mathbf{5}\rangle_B \\
			& \left|\phi_2\right\rangle_{A B}=|\mathbf{2}\rangle_A|\mathbf{1}-\mathbf{2}+\mathbf{5}-\mathbf{3}\rangle_B \\
			& \left|\phi_3\right\rangle_{A B}=|\mathbf{1-2}\rangle_A|\mathbf{0}-\mathbf{4}\rangle_B \\
			& \left|\phi_4\right\rangle_{A B}=|\mathbf{0-1}\rangle_A|\mathbf{2}-\mathbf{3}\rangle_B \\
			& \left|\phi_5\right\rangle_{A B}=|\mathbf{0+1+2}\rangle_A|\mathbf{0}+\mathbf{1}+\mathbf{2}+\mathbf{3}+\mathbf{4}+\mathbf{5}\rangle_B 
		\end{aligned}
		\label{5}
	\end{multline}
	It is quite straightforward to show that the set $\mathcal{S}_2$ considered above is free from local redundancy \cite{bandhyopadhyay201,Li2022,subrata2024}. Here, Bob's system can be considered to be the composition of qubit and qutrit subsystems. Precisely, $|\mathbf{0}\rangle_B:=|00\rangle_{b_1 b_2},|\mathbf{1}\rangle_B:=|01\rangle_{b_1 b_2},|\mathbf{2}\rangle_B:=$ $|02\rangle_{b_1 b_2},|\mathbf{3}\rangle_B:=|10\rangle_{b_1 b_2},\left|\mathbf{4}_B\right\rangle:=|11\rangle_{b_1 b_2},|\mathbf{5}\rangle_B:=|12\rangle_{b_1 b_2}$. Take two states, $\left|\psi_3\right\rangle_{A B}$ and $\left|\psi_4\right\rangle_{A B}$. When any of the subparts (qubit or qutrit) of Bob's system for both states is discarded the reduced states will be nonorthogonal. Similar things happen for Alice also. This implies the set $\mathcal{S}_2$ is free from local redundancy.
	\par Now we will show that the set $\mathcal{S}_2$ is locally distinguishable. The players can avail the following discrimination protocol. First Bob performs a measurement $\mathbb{M}_B \equiv\left\{\mathbb{M}_B^1:=\right.$ $P\left[|\mathbf{0}-\mathbf{4}\rangle_B\right], \mathbb{M}_B^2 :=P\left[|\mathbf{2}-\mathbf{3}\rangle_B\right], \mathbb{M}_B^3:=P[\mid \mathbf{0}+\mathbf{1}+\mathbf{2}+\mathbf{3}+\mathbf{4}$ $\left.\left.+\mathbf{5}\rangle_B\right], \mathbb{M}_B^4:=\mathbb{I}-\left(\mathbb{M}_B^1+\mathbb{M}_B^2+\mathbb{M}_B^3\right)\right\}$. Here, $P\left[|\cdot\rangle\right]:=$ $|\cdot\rangle\langle\cdot|_{\mathcal{P}}$, and $\mathcal{P}$ denotes the party. When $\mathbb{M}_B^1$ clicks, the given state must be $\left|\phi_3\right\rangle$. Similarly, for the click $\mathbb{M}_B^2$, the state is $\left|\phi_4\right\rangle$. Also for the outcome $\mathbb{M}_B^3$ the isolated states are $\left|\phi_5\right\rangle$. Whenever $\mathbb{M}_B^4$ clicks the given state can be $\left|\phi_1\right\rangle$, $\left|\phi_2\right\rangle$. However, in that case, Alice can perform a measurement $\mathbb{M}_A \equiv\left\{\mathbb{M}_A^1:=\right.$ $P\left[|\mathbf{0}\rangle_A\right],$ $\mathbb{M}_A^2 := \mathbb{I} - \mathbb{M}_A^1\}$
	  to distinguish between these two states. This concludes the local discrimination protocol for the set $\mathcal{S}_2$. In the following, we will demonstrate a protocol to activate nonlocality without entanglement from this set.
	\begin{theorem} \cite{GhoshStrongActivation2022}
		The set $\mathcal{S}_2$ is a locally distinguishable set and can be transformed deterministically to a locally irreducible set via LPCC.  
	\end{theorem}
	\begin{proof} 
		Consider that Bob performs a local measurement $\mathcal{K}_B \equiv\left\{\mathcal{K}^B_1:=P\left[(|\mathbf{0}\rangle,|\mathbf{1}\rangle,|\mathbf{2}\rangle)_B\right], \mathcal{K}^B_2:=\right.$ $\left.P\left[(|\mathbf{3}\rangle,|\mathbf{4}\rangle,|\mathbf{5}\rangle)_B\right]\right\}$, $P\left[(|i\rangle,|j\rangle,\dots)_{\mathcal{P}}\right] = \left[(|i\rangle\langle i|+|j\rangle\langle j|+\dots)_{\mathcal{P}}\right]$, $\mathcal{P}$ stands for party. If $\mathcal{K}^B_1$ clicks, they end up with,
		\begin{figure}[h!]
			\centering
			\includegraphics[width=0.48\textwidth]{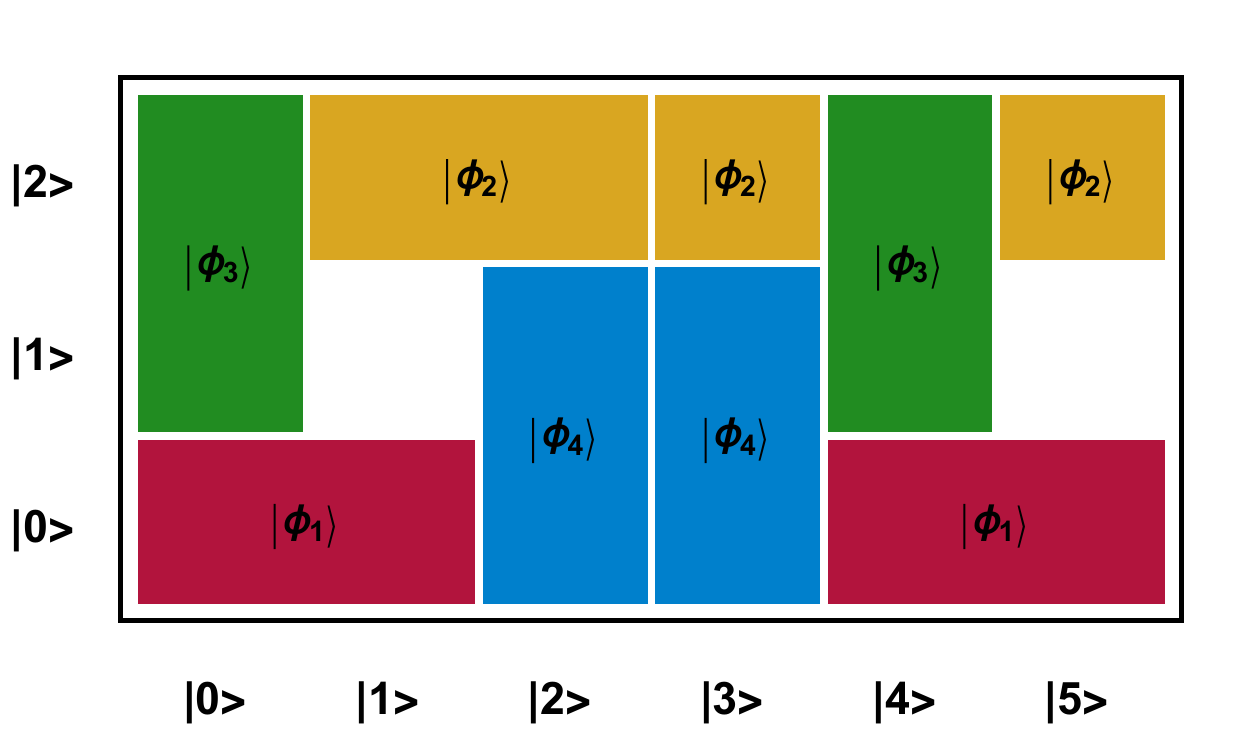}
			\caption{Tiling diagram for the states in \( \mathbf{\mathcal{S}_2} \). The outlined region indicates the support of Bob's measurement outcomes \( \mathcal{K}_1^B \) and \( \mathcal{K}_2^B \) resulting in post-measurement states are all contained in a UPB subspace.}\label{fig44-4}
		\end{figure}
		$$
		\left\{\begin{array}{c}
			|\mathbf{0}\rangle_A|\mathbf{0}-\mathbf{1}\rangle_B,|\mathbf{2}\rangle_A|\mathbf{1}-\mathbf{2}\rangle_B, |\mathbf{1-2}\rangle_A|\mathbf{0}\rangle_B,\\
			|\mathbf{0-1}\rangle_A|\mathbf{2}\rangle_B, 
			|\mathbf{0+1+2}\rangle_A|\mathbf{0}+\mathbf{1}+\mathbf{2}\rangle_B,
		\end{array}\right\}
		$$
		which is a locally irreducible set \cite{BennettUPB1999}. On the other hand, if Bob gets $\mathcal{K}^B_2$, they are then left with the following states,
		\[	\left\{\begin{array}{c}
			|\mathbf{0}\rangle_A|\mathbf{4}-\mathbf{5}\rangle_B,|\mathbf{2}\rangle_A|\mathbf{5}-\mathbf{3}\rangle_B, |\mathbf{1-2}\rangle_A|\mathbf{4}\rangle_B,\\
			|\mathbf{0-1}\rangle_A|\mathbf{3}\rangle_B, 
			|\mathbf{0+1+2}\rangle_A|\mathbf{3}+\mathbf{4}+\mathbf{5}\rangle_B,
		\end{array}\right\}
		\]
		which is a locally irreducible set. It is clear that the five post-measurement states for each cases of Bob's measurement $\mathcal{K}^B_1,\;\mathcal{K}^B_2$ form the celebrated unextendable product basis (UPB) \cite{BennettPB1999,BennettUPB1999} in $\mathbf{C^3 \otimes C^3}$, see Fig.~\ref{fig44-4}. It has been well established that UPB is locally indistinguishable \cite{Divincinzo,BennettUPB1999}. So, the set $\mathcal{S}_3$ is activable by LPCC. Hence, this completes the proof.
	\end{proof}
\end{document}